\documentclass[journal]{IEEEtran}
\usepackage{cite}
\usepackage[pdftex]{graphicx}
\usepackage{graphicx}
\usepackage{endnotes}
\usepackage{amsmath}
\usepackage{amsthm}
\usepackage[top=54pt,left=54pt,right=54pt,bottom=54pt]{geometry}
\usepackage{amssymb}
\usepackage{color}
\usepackage{verbatim}
\usepackage[latin9]{inputenc}

\newtheorem{lem}{Lemma}

\newtheorem{rem}{Remark}

\title{\LARGE \bf 3D Mobile Localization Using Distance-only Measurements}
\author{Bomin Jiang, Brian D.O. Anderson and Hatem Hmam
\thanks {*This work is supported by National ICT Australia, which is funded by the Australian Research Council through the ICT Centre of Excellence program. Brian D.O. Anderson is also supported by Australian Research Council under grant DP160104500.}
\thanks{Bomin Jiang is with Institute for Data, Systems, and Society, Massachusetts Institute of Technology, Cambridge, MA, USA. Brian D. O. Anderson is with Hangzhou Dianzi University, China, Data-61 CSIRO and Research School of Engineering, The Australian National University, Canberra, Australia. Hatem Hman is with Defence Science and Technology Group, Edinburgh, Australia. }
\thanks{\tt\small{bominj@mit.edu, Brian.Anderson@anu.edu.au, Hatem.Hmam@dst.defence.gov.au.}}
}

\begin{document}
\maketitle
\thispagestyle{empty}
\pagestyle{empty}

\begin{abstract}
For a group of cooperating UAVs, localizing each other is often a key task. This paper studies the localization problem for a group of UAVs flying in 3D space with very limited information, i.e., when noisy distance measurements are the only type of inter-agent sensing that is available, and when only one UAV knows a global coordinate basis, the others being GPS-denied. Initially for a two-agent problem, but easily generalized to some multi-agent problems, constraints are established on the minimum number of required distance measurements required to achieve the localization. The paper also proposes an algorithm based on semidefinite programming (SDP), followed by maximum likelihood estimation using a gradient descent initialized from the SDP calculation. The efficacy of the algorithm is verified with experimental noisy flight data. 
\end{abstract}

\section{introduction}

Research on multi-agent systems has been attracting increasing interest
in recent years. This is due to the potential applications of this
research in a variety of fields, including sensor networks \cite{mao2007wireless},
distributed power grids \cite{nagata2002multi}, distributed computation
and so on. For certain multi-agent systems, such as a cooperating
group of unmanned airborne vehicles (UAVs), there are some tasks,
e.g. formation shape control \cite{jiang2015simultaneous}, state
consensus \cite{cao2008reaching}, or self-localization \cite{mao2007wireless},
which often need to be performed in a distributed way and using limited
sensing.

A common task for UAV formations is that of localizing a target, such as a radio-frequency emitter. A prerequisite for this is for all agents to know their positions in a common (global) coordinate basis, since otherwise vehicles cannot sensibly cooperate by sharing measurements from the emitter in order to localize it. However, not all vehicles may have access to GPS, depending on their location or possible jamming, but if one vehicle does have such access, then it can seek to localize the other vehicles
in the same formation (allowing subsequently all vehicles to localize the target). Those other vehicles, while not GPS-equipped, are assumed to be equipped with inertial navigation sensing, which means that each views its position and motion in a local coordinate frame whose orientation and relative to the global frame are unknown. The book chapter \cite{chapter9} provides a high level introduction to localisation in a 3D environment, including limited consideration of mobility. While this can be achieved in principle
if vehicles have sensors delivering relative positions (range and bearing)
of their neighbors, and an ability to communicate with their neighbors,
distance-only sensing is preferred in many scenarios \cite{jiang2015simultaneous,cao2011formation}
due to its reliability, low cost, low power and light weight sensors. Therefore,
a method that can localize non-GPS-equipped UAVs 
\textit{using distance-only
measurements} (and INS data from non-GPS-equipped vehicles) is of great appeal.  Providing such a method is the motivation for this paper. Additionally, we seek  a method which degrades gracefully, rather than collapses, in the presence of increasing noise.

For the sake of completeness, we remark that there are a number of works using very different approaches in the utilization of distance-only measurements in multi-agent systems to achieve localization of vehicles whose global coordinates are not measured by GPS; most of them require the vehicles, playing the role of mobile sensors, to move in certain \textit{standard trajectory patterns} \cite{jiang2013translational}.
The trajectories are adopted simply for the purpose of achieving self localization in a global coordinate basis
of each agent, and so may conflict with achieving formation objectives
such as surveillance of a target. 
Fundamentally too, this setting, though presented in many papers, is likely to be energy inefficient.
In another direction again, some work in the robotics literature studies
the problem of determining relative position and orientation between
two robots \cite{trawny20073d} given distance measurements between them. The paper \cite{trawny20073d} does
some pioneering work in this area and investigates this problem in
3D, with interagent distances assumed to be measurable with high though
not perfect accuracy. However, this method is only applicable to agent
pairs, not to multi-agent formations. Further, noise robustness is not high. 
In our paper, we obtain location information, as long as orientation information of the local coordinates using distance-only measurements, much like localizing a rigid body \cite{chepuri2014rigid}. 

On the algorithm side, there are some previous work looking at localization problems using SDP, see e.g. \cite{shamsi2010sensor}. However, those formulations cannot be adopted in the case of distance-only measurements. In this case, the relative position vector must be squared to reveal the distance-only measurements, and then, in the noisy case, the associated least-squares formulation will lift the order of polynomials to 4. Our paper gives a SDP formulation specifically used in the distance-only measurements case that remains of order 2. 

This paper provides an implementation of a localization
algorithm given noisy distance measurements combining semidefinite
programming (SDP) and Maximum Likelihood Estimation (MLE),
which is applicable in both a two-agent situation and a number of multi-agent
situations. 
The two-agent situation is particularly important because, in a multi-agent setting, it is often desired that only one agent emitting signals is enough for the localization of all other agents in the network. Such a scenario gives a star topology, which is trivial to treat once we obtain solutions in the two-agent situation. However, we will look at generic multi-agent cases in the future. 
Compared to other similar research, in our proposed method,
1. Agents are not required to move in any designated pattern; 2.
Only one GPS equipped agent (which must be moving) is required; 3.
Only distance measurements between sensors are used, and 4. real noisy data verifies the
robustness of the algorithm in noise. In short, this paper provides
a practical SDP formulation to achieve localization using distance-only measurements. 

The rest of this paper is organized as follows: 

Section II poses, in a mathematical framework, the problem of localizing in 3D space an agent without GPS using distance measurements to a second GPS-equipped agent. It also discusses the minimal number of measurements required via the theory of graph rigidity. Section III is the most comprehensive section of the paper; it explains how the problem can be solved, using several major steps. These steps include semidefinite programming (SDP), which gives an initial solution, and determination of a maximum likelihood estimate via a gradient algorithm initialized by the SDP solution. Issues of gradient flow on a manifold are also discussed. In addition, the usage of the proposed algorithm in limited multi-agent cases is also discussed. Simulations and separately results from flight testing data are presented in Section IV. Concluding
remarks are provided in Section VI.

\section{The 3D localization problem and number of measurements required}

In this section, we first introduce the general
localization problem for an agent pair. Then we use graph rigidity theory to determine a minimal number of measurements required to achieve localization. 

\subsection{Problem formulation}
The problem of a pair of agents localizing each
other is the most elementary version of the larger scale multi-agent
localization problem. Consider two agents in 3D space, as shown in
Figure \ref{fig1-1}. Agent 0 is the reference agent, whose position
in the global coordinate system is denoted by $p_{0}=[x_{0},y_{0},z_{0}]^{\top}$.
Agent 1 is the agent to be localized, whose position in its own local
coordinate system (almost certainly derived from an inertial navigation system) is denoted by $p_{1}=[x_{1},y_{1},z_{1}]^{\top}$.
Suppose further that a coordinate system transformation, including
a $3\times3$ rotation matrix $R$ and a $3\times1$ translation vector
$T$, can transfer agent 1's local coordinates into global coordinates,
so that the position of agent 1 in the global coordinate system is 
$Rp_{1}+T$. The matrix $R$ and vector $T$ are assumed to be constant
(and unknown) through the interval over which measurements are taken, i.e. the cumulative drift in the INS system over the time interval of measurements is negligible. Equivalently, the time interval is sufficiently short. 

\begin{figure}
\begin{centering}
\includegraphics[width=0.42\textwidth]{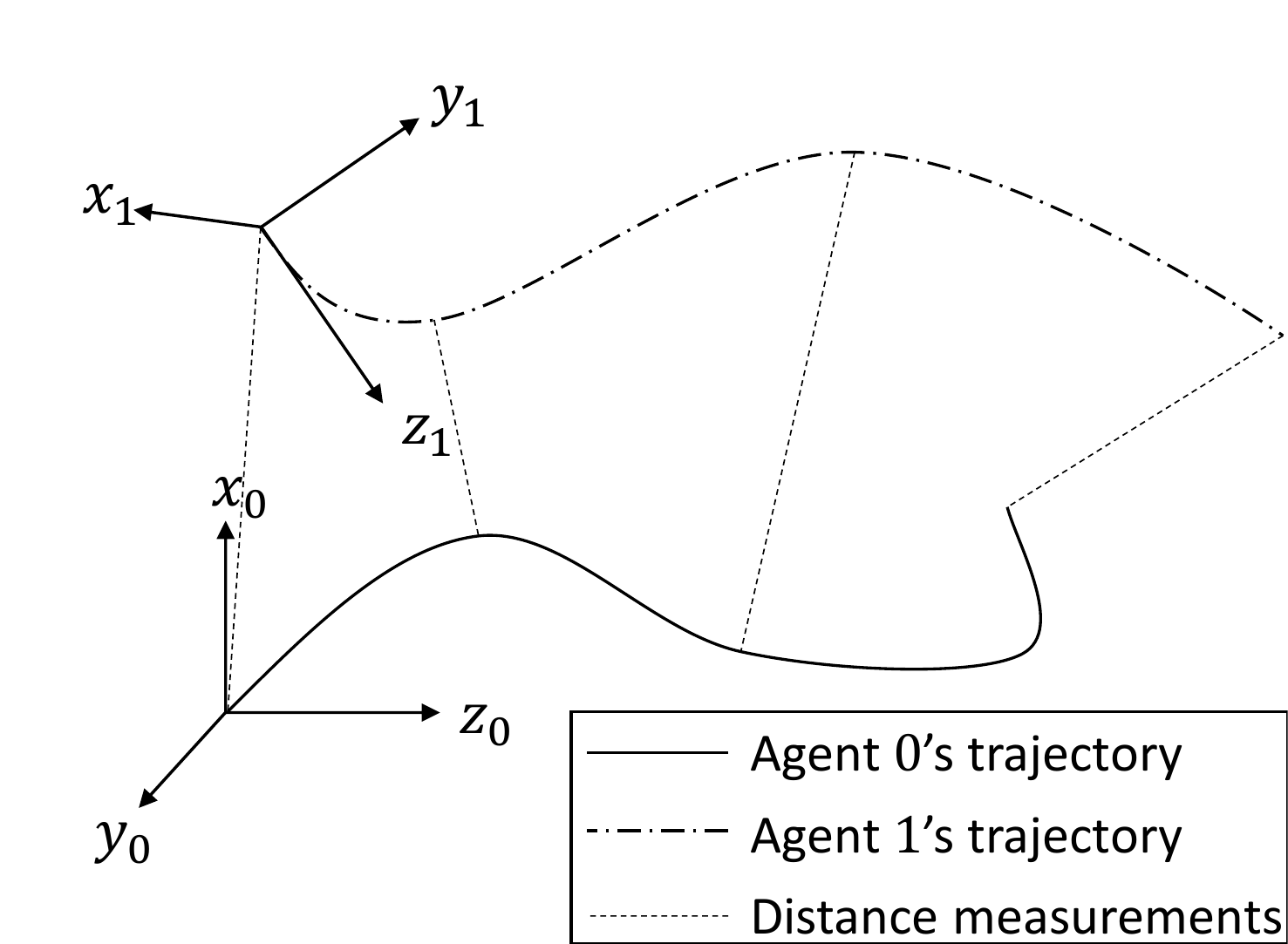}
\par\end{centering}
\caption{Demonstration of the elementary two-agent localization problem}
\label{fig1-1}
\end{figure}

In addition, suppose $d$ is the distance between
the two agents and define 
\[
\bar{p}=Rp_{1}+T-p_{0}
\]
Thus $\bar p$ is the relative position vector
of agent 1 as seen from agent 0, using global coordinates. Immediately
we have 
\begin{equation}
d=\|\bar{p}\|=\|Rp_{1}+T-p_{0}\|\label{distance norm-1}
\end{equation}

Suppose at time instants  $\tau=\tau_{k},~~k=1,2,\cdots,N$,
each agent takes distance measurements $d[k]$ and each knows the coordinates
$p_{0}[k]$ and $p_{1}[k]$ at these time instants through messsage interchange. The objective
is to infer $R$ and $T$ after a finite number of measurements. Knowledge of $R$ and $T$ yields knowledge of the positions of agent 1 in global coordinates, i.e. solves the localization problem for that agent.

Note that there are 6 degrees of freedom to pin down, three each associated with $R$ and $T$. Therefore, one must expect that at least 6 measurements would be needed to solve the relevant equation set. As summarized in \cite{dietmaier1998stewart} there are generally 40 solutions of that set 
\footnote{The equation set is actually a set of simultaneous polynomial equations. Unless all the equations are linear, as a matter of algebra n simultaneous equations in n unknowns always have multiple solutions, provided they are solvable in the first place. }. 
One disambiguates only by using further measurements. 
Although for generic agent trajectories, the minimal number of measurements to secure disambiguation is 7, there could be a problem if the trajectories are special. E.g., for two agents both flying in exactly straight lines, as it turns out, distance only measurements will never give a unique solution. In the rest of this paper, we refer to a "generic trajectory" as a trajactory that has a unique solution when there are more than 7 distance measurements. Note the set of "generic trajectories" will happen with probability 1 if agents are doing a random walk. 
\footnote{The situation is analagous to what happens when one is solving the linear equation Ax=b for square A. Generic A are those which are nonsingular. Nongeneric or special A are those which are singular. Randomly selected A will with probaiblity 1 be nonsingular or generic. Just as linear equations may not be solvable for certain parameter values, so can the nonlinear equations relevant here be nonsolvable for certain parameter values, essentailly defined by special trajectories. }  

Regarding the SDP based approach we discuss later in this paper, if less than 7 measurements are taken, then the algorithm will randomly give one solution from the multilple possibilities that satisfies all measurements, but the solution may be far away from the ground truth. 
Noise contaminating the distance measurements is a further issue to contemplate.

\subsection{Identifying the minimal number of measurements required}

On the theory side, there are some previous papers providing some
analysis of this problem for the two agent case or solving a simplified version of this problem
in 2D. The paper \cite{motevallian2013new} points out that the two-agent
problem in 2D is almost always solvable with no less than 4 inter-agent
measurements, followed by a computationally effective solution of
the problem in a two-dimensional ambient space via the concept of
four bar chain in mechanics \cite{Chung2005coupler}. Because the
solution of this two-agent localization problem in 3D is analogous
to solving the forward kinematics of a Stewart-Gough platform \cite{stewart1966plantform}
without the constraints that both the fixed and mobile platform are
planar, we know that with 6 distance measurements, the elementary
two-agent localization problem stated formally in Section \ref{sec:A tentative solution} normally has 40 distinct solutions \cite{dietmaier1998stewart}, though some of the
solutions may not be real and hence must be discarded. Not unsurprisingly, an additional, i.e. 7th,  interagent distance measurement is enough to disambiguate the multiple solutions, at least generically, so that a unique solution results. 

A related result along these lines can be found in \cite{yu2006principles}, expressed in terms of the merging of globally rigid  frameworks. A globally rigid  framework essentially is a framework of vertices and  bars of known lengths joining the vertices, such that the specification of the lengths defines the framework up to a congruence. Now the coordinate-alignment problem of interest to us can generally be reformulated as one which (a) associates a 3D (globally rigid) body with each agent (the body being
a polyhedron with vertices defined by the agent positions at each
time an interagent distance is measured), (b) associates a bar (or
link between two bodies) with each interagent distance measurement,
and (c) seeks conditions on the overall number of links and rules
concerning the points on the bodies to which they can be incident
to ensure that the entire  framework is globally rigid,
i.e. every joint or bar end in the framework has a computable location
in a coordinate frame fixed to the framework, which cannot flex. 

The result of \cite{yu2006principles} states that if two globally rigid  frameworks in an ambient three-dimensional space are joined by 7 or more bars of fixed lengths (and the endpoints of the bars can only coincide to a limited extent, or are distinct) then the merged framework is again globally rigid. This means that if one of the frameworks before merging is localized, the whole framework after merging (and thus the second framework before merging) is localized.

Other relevant literature includes that on {\it{multi body-bar problems}}, see \cite{tay1984recent,connelly2013generic}. A body-bar framework is one where bars are regarded as joining rigid bodies, rather than intersect in a point or joint. In our context, each rigid body is defined, as noted above, by the set of points at which a particular agent is located at the times distance measurements are taken.  These papers focus on determination of counting conditions on the links to ensure rigidity (nondeformability roughly) and global rigidity (unique  specification up to congruence) of body-bar frameworks. In the particular case of two agents, they provide the same conclusion already noted, i.e. generically 7 measurements are required. The counting conditions are also consistent with the limited set of multiagent cases we discuss later in this paper. 

Hence, in the end, we consider vehicle localization for set-ups corresponding to a proper subset of all possible globally rigid body-bar frameworks, or equivalently a subset of all possible information patterns linking multiple UAVs, in order that recursive localization can be achieved. 

\section{Solution to the two-agent problem}

\label{sec:A tentative solution} In this section, we consider the
two agent case, apart from noting at the end of the section in a Remark how the ideas can be applied to certain multiagent arrangements. We first present, in subsection A, a semidefinite programming (SDP)
approach for solving a constrained least squares problem in the case
of no less than 7 measurements. This solution actually is a relaxed solution (in the sense of optimization theory) of the problem of interest, and so a further step may have to be undertaken to deal with possible consequences of the relaxation. This further step may induce failure of
 the constraint that the matrix $R$ is orthogonal with
determinant 1, and so we show how to modify the outcome of the calculations
of Subsection 3.1 to achieve an orthogonal matrix in Subsection 3.2. The outcome of the SDP  computations is an estimate of the rotation matrix and translation vector. These estimates however are not maximum likelihood estimates, even if the noise contaminating the distance measurements is zero mean and gaussian. In Subsection 3.3, we show how a maximum likelihood estimate of the matrix and vector can be obtained by a gradient flow; it is vital to initalize this properly, and the result of the SDP calculation is used for that purpose.

In this section, because there are only two agents, we let agent $x$ be the reference agent, whose position in the global coordinate
system is denoted by $p_{x}=[x_{1},x_{2},x_{3}]^{\top}$. On the other
hand, agent $y$ is the agent to be localized, whose position in
its own local coordinate system is denoted by $p_{y}=[y_{1},y_{2},y_{3}]^{\top}$.
 Suppose further that a coordinate system transformation, including
a $3\times3$ rotation matrix $R$ and a $3\times1$ translation vector
$T$, can transfer agent $y$'s local coordinates into global coordinates,
so that the position of agent $y$ in the global coordinate system
is $\bar{p}=Rp_{y}+T$. As before, $d$ denotes an interagent distance. 

\subsection{A semidefinite programming (SDP) relaxation}

\label{sec::A linear processing approach} Suppose $R=\{r_{ij}\}$,
$T=[t_{1},t_{2},t_{3}]^{\top}$. 
There holds for measurements at time $\tau_{1},\tau_{2},\dots \tau_{N}$ (with the time index suppressed in the following equations)

\begin{equation}
d^2=\|\bar p_y-p_x\|^2\label{eq::display z2}
\end{equation}

The equation for $d^{2}$ can be written as

\begin{equation}\label{eq:newshort z2}
\begin{split}
d^2&=\|\bar p_y-p_x\|^2\\
&=\|Rp_y+T-p_x\|^2\\
&=\|p_y\|^2+\|T\|^2+\|p_x\|^2+2T^{\top}Rp_y-2p_x^{\top}Rp_y-2p_x^{\top}T
\end{split}
\end{equation}

The fourth term involves quadratic expressions of the unknowns, i.e. the entries of $R$ and $T$. The fifth and sixth terms involve linear terms only. The end result is the following quadratic constraint on the entries of $R$ and $T$,  presented to place all the unknowns on the right side of the equation:

\begin{equation}
\begin{split} & d^{2}-\|p_{x}\|^{2}-\|p_{y}\|^{2}\\
= & -2x_{1}y_{1}r_{11}-2x_{1}y_{2}r_{12}-2x_{1}y_{3}r_{13}\\
- & 2x_{2}y_{1}r_{21}-2x_{2}y_{2}r_{22}-2x_{2}y_{3}r_{23}\\
- & 2x_{3}y_{1}r_{31}-2x_{3}y_{2}r_{32}-2x_{3}y_{3}r_{33}\\
- & 2x_{1}t_{1}-2x_{2}t_{2}-2x_{3}t_{3}\\
+ & 2y_{1}\sum r_{i1}t_{i}+2y_{2}\sum r_{i2}t_{i}+2y_{3}\sum r_{i3}t_{i}+\sum t_{i}^{2}
\end{split}
\label{eq::short z2}
\end{equation}

If we regard each summand in \eqref{eq::short z2} as a product of known values $-2x_{1}y_{1},~-2x_{1}y_{2},~\cdots,~1$ and independent unknowns $r_{11},~r_{12},~\cdots,~\sum t_{i}^{2}$, one can seek to solve the set of equations with a sufficient number of measurements. Generally 16 measurements are sufficient because there are 16 linearly independent unknowns $r_{11},~r_{12},~\cdots,~\sum t_{i}^{2}$ (at least if the associated coefficient matrix has full rank, which proves to be almost always the case). When the coefficient matrix is close to singular, more measurements are required. If there are more measurements than required, a least-squares solution is used.

In the above solution process, we regard each summand in \eqref{eq::short z2}
as a product of a known value and an independent unknown. This therefore
leaves out the consideration of the nonlinear constraints (though
recall that \eqref{eq::short z2} only arises after use of a number
of such constraints). Using these constraints may reduce the number
of measurements required and improve the estimation accuracy in the
presence of noise as we now argue in more detail. 

Define $\Theta=[\theta_{1},\theta_{2},\cdots,\theta_{16}]^{\top}$
to be the 16-vector of unknowns $r_{11},~r_{12},~\cdots,\sum r_{i3}t_{i},~\sum t_{i}^{2}$
in \eqref{eq::short z2} and $A[k]=[a_{k~1},a_{k~2},\cdots,a_{k~16}]$
to be the row vector of known values $-2x_{1}y_{1},~-2x_{1}y_{2},~\cdots,~1$
in \eqref{eq::short z2} at time $t=t_{k}$. Further, assuming for convenience there are 16 measurements,  suppose $A=\Big[A[1]^{\top}~A[2]^{\top},\cdots,A[16]^{\top}\Big]^{\top}$
and $b$ is a column vector with the $k$th entry being $z[k]^{2}-\|p_{x}[k]\|^{2}-\|p_{y}[k]\|^{2}$.

The optimization problem we are trying to solve is

\begin{equation}
\begin{split}\mbox{argmin}_{\Theta}~~\frac{1}{2}\|A\Theta-b\|^{2}\\
\mbox{subject to}~~C(\Theta)=0
\end{split}
\label{eq::constratint optimization}
\end{equation}
where $C(\Theta)=0$ expresses all the constraints among the $\theta_{i}$.
Note there are altogether 10 independent constraints expressed in $C(\Theta)$
as listed in the Appendix. An additional four constraints are listed which are not independent of the 10 just mentioned. 

\textcolor{black}{We now show how to solve \eqref{eq::constratint optimization},
using semidefinite programming. Define 
\[
X=\left[\begin{array}{c}
\Theta\\
-1
\end{array}\right]\left[\begin{array}{c}
\Theta\\
-1
\end{array}\right]^{\top}
\]
}

\textcolor{black}{Solving for $\Theta$ is equivalent to solving for
$X$ with the constraints that}

\textcolor{black}{1. $X$ is positive semidefinite (denoted by $X\succcurlyeq0$)}

\textcolor{black}{2. $\mbox{rank}(X)=1$}

\textcolor{black}{3. The bottom right corner element $X_{17,17}=1$}

\textcolor{black}{Suppose that $\left\langle u,v\right\rangle $ denote
the inner product of two matrices $u$ and $v$, i.e. $\left\langle u,v\right\rangle =\mbox{trace}(u^{\top}v)$,
and define $P=\left[\begin{array}{cc}
A & b\end{array}\right]^{\top}\left[\begin{array}{cc}
A & b\end{array}\right]$}

The objective function can be written as

\[
 \frac{1}{2}\|A\Theta-b\|^{2}
\propto \left[\begin{array}{c}
\Theta\\
-1
\end{array}\right]^{\top}\left[\begin{array}{cc}
A & b\end{array}\right]^{\top}\left[\begin{array}{cc}
A & b\end{array}\right]\left[\begin{array}{c}
\Theta\\
-1
\end{array}\right]
=  \left\langle P,X\right\rangle 
\]

Similarly, the constraints in \eqref{eq::constraints}
can be written in terms of $X$

\[
\left\langle Q_{i},X\right\rangle =q_{i},i=1,\cdots,10
\]
where $Q_{i}$ are 17 by 17 symmetric matrices and $q_{i}$ are 10 scalars. Now the optimization problem becomes
\begin{equation}
\begin{split}\mbox{argmin}_{\Theta} & \left\langle P,X\right\rangle \\
\mbox{subject to} & \left\langle Q_{i},X\right\rangle =q_{i},i=1,\cdots,10\\
 & X_{17,17}=1\\
 & X\succcurlyeq0\\
 & \mbox{rank}(X)=1
\end{split}
\label{eq:SDP}
\end{equation}

A naive semidefinite programming (SDP) relaxation
is given by removing the rank constraint. After removing the constraint,
$C_{i}=0,i=1,\cdots,10$ no longer implies $C_{i}=0,i=11,\cdots,14$, so it is best to add those constraints as well. As observed in a number of simulations, the above additional constraints are helpful to obtain a low-rank solution for generic trajectories, and therefore helpful to reduce the rounding error when we recover the rank-1 solution using SVD (as shown later in the paper). In addition, because entries in the rotation matrix should satisfies $-1\leq r_{ij} \leq 1$, one can also consider including Reformulation-linearization-technique (RLT) constraints to further reduce the set of feasible solutions. 

After finding the solution, we can find its best
rank 1 approximation under matrix 2-norms by using SVD and setting
all but the largest singular value to zero, see Section 4.3 of \cite{hopcroft2012computer}.
In fact, it is generally a very good approximation because the solution
of the relaxed problem is very close to rank 1. In a number of numerical
simulations, the largest singular value of the solution is generally
$10^{2}\thicksim10^{5}$ larger than the second largest singular value. 

\textcolor{black}{The solution of the above SDP can then be used as
an \textit{initial condition} to the gradient optimization presented below in
Subsection \ref{sec::Gradient Descent Method}. One should note that
the number of measurements used in the semidefinite programming
method can increase arbitrarily and decrease to 7. Although the validity
of the semidefinite programming approach does not straightforwardly
imply that the minimum number of measurements is 7, we nevertheless
assume that this constraint always holds; this will be discussed in more
detail in Subsection \ref{sec::number of measurements}. }

\subsection{Obtaining a rotation matrix}

\label{sec::Procu Correction} In the noisy case, let $\tilde{T}$
and $\tilde{R}$ denote the estimated value of $T$ and $R$ respectively.
Note that the imposition of the rank 1 constraint through approximation,
as a final tidy-up step of the algorithm, may destroy orthogonality,
though up to that point, SDP guarantees orthogonality by virtue of
the equality constraints. In this case, the obtained $\tilde{R}$
may not satisfy all the conditions to be a rotation matrix; therefore,
one more step can be taken to find the rotation matrix $\bar{R}$
that has the closest Frobenius norm to $R$. Thus one seeks
\begin{equation}
\bar{R}=\mbox{arg}\min_{\Omega}\|\Omega-R\|_{F}~~~~~\mbox{subject to }\Omega\Omega^{\top}=I, \det \Omega=1.\label{Orthogonal Procrustes problem}
\end{equation}
where $\|\cdot\|_{F}$ denotes the Frobenius norm.

This minimization problem is a special case of the Orthogonal Procrustes problem
\cite[pp. 29-34]{gower2004procrustes}. To find this orthogonal matrix
$\bar{R}$, the singular value decomposition can be used 
\begin{equation}
\tilde{R}=U\Sigma V^{*}\label{eq:svd of R}
\end{equation}

Suppose $J$ is a diagonal matrix
with the last entry on the diagonal being $-1$ and all other entries
on the diagonal being $1$. The solution of this constrained version
of the Orthogonal Procrustes problem is

\begin{equation}
\bar{R}=\left\{ \begin{array}{cc}
UV^{*} & \mbox{if}~~\mbox{det}(UV^{*})=1\\
UJV^{*} & \mbox{if}~~\mbox{det}(UV^{*})=-1
\end{array}\right.\label{eq:solution of specific procrustes}
\end{equation}

Note that det$(\Sigma-I)$ can be used as a error measure.

\subsection{Maximum Likelihood Estimation and Gradient Flow on Manifold}

\label{sec::Gradient Descent Method} To refine the solution of the
SDP, it is natural to contemplate a gradient descent optimization.
In the noisy situation, we want a maximum likelihood estimate of $R$
and $T$, and this will not be given by the quadratic index and constrained
least squares estimate of SDP. Hence we will formulate a new index
whose minimization yields the MLE, and then examine a gradient descent
algorithm to compute the minimum. 

Suppose the measurement of distance $z[k]$ is contaminated by a Gaussian
noise with zero mean and some standard deviation; i.e. the measurement sensor delivers $\tilde{z}[k]=z[k]+\xi,~\xi\sim N(0,\sigma^{2})$.
We assume in this section that $p_{x}$ and $p_{y}$ can be obtained
without noise and that the measurement noise values at different times
are independent.

Now we obtain $\tilde{z}[k]-\|Rp_{y}[k]+T-p_{x}[k]\|\sim N(0,\sigma^{2})$.
The likelihood function is

\begin{equation}
\begin{split} & \mathcal{L}(p_{x},p_{y},z|R,T)\\
= & \frac{1}{\sigma\sqrt{2\pi}}\prod_{k=1}^{N}\mbox{exp}\Big[-\frac{(\tilde{z}[k]-\|Rp_{y}[k]+T-p_{x}[k]\|)^{2}}{2\sigma^{2}}\Big]
\end{split}
\label{MLE}
\end{equation}
and 
\begin{equation}
\begin{split} & \mbox{log}\mathcal{L}(p_{x},p_{y},z|R,T)\\
= & \sum_{k=1}^{N}\Big[-\frac{(\tilde{z}[k]-\|Rp_{y}[k]+T-p_{x}[k]\|)^{2}}{2\sigma^{2}}\Big]-\mbox{log}(\sigma\sqrt{2\pi})
\end{split}
\end{equation}

Therefore the maximum likelihood estimate is given by solving the
optimization problem below

\begin{equation}
\begin{split}\tilde{R},\tilde{T}=\mbox{arg}\min_{R,T}\sum_{k=1}^{N}(\tilde{z}[k]-\|Rp_{y}[k]+T-p_{x}[k]\|)^{2}\\
\mbox{subject to }RR^{\top}=1~~\mbox{det}(R)=1
\end{split}
\label{optimization problem}
\end{equation}

Like many MLE estimation problems, the index is not convex and minimization
is not necessarily straightforward. A particular problem with using gradient descent on any nonconvex function is to find an initialization within the capture region of the global minimum, and this is where the calculations of the previous subsections become relevant if not critical. We use the
result from linear processing (with the number of measurements being
16) or the SDP approach (with the number of
measurements being greater than or equal to 7) in the previous subsections
to initialize  a gradient descent algorithm aimed at
finding the minimum. 
\footnote{In unpublished work \cite{you2017underwater} studying  autonomous underwater vehicle localization using distance-only measurements to a non-GPS-equipped vehicle, minimization of the same MLE index is tackled using a form of approximation for the index, with application of a parallel projection algorithm. Proper comparison with the methods of this paper cannot be made, due to the limited details in the paper. How issues of initialization can be effectively tackled is also not clear. }

It is useful for this purpose to know how to calculate the gradient
of a function of a special orthogonal matrix on the manifold of special
orthogonal matrices. Consider $f$: $SO_{3}\rightarrow\mathbb{R}$,
mapping special orthogonal matrices to the reals. Suppose we want
to compute the gradient, reflecting the orthogonal property. 

The general idea (technically a consequence of the fact that $SO_{3}$ is a Riemannian manifold and so inherits a metric from the Euclidean
space in which it is embedded \cite{helmke1994optimization}) is:
first we consider a point $R\in\mathbb{R}^{3\times3}$ on the $SO_{3}$
manifold, and compute the gradient $\frac{\partial f}{\partial R}$ in the standard way, then we project it onto the tangent space of $SO_{3}$ at the point $R$.
For this purpose, we need to have the tangent space and normal spaces
of $SO_{3}$ at some point $R$, and also the projection of a vector
to the tangent space.

\begin{lem} i. The tangent space of $SO_{3}$ at $R\in SO_{3}$
is the set of $\mathcal{P}$ such that \cite{helmke1994optimization}
\[
\mathcal{P}^{\top}R+R^{\top}\mathcal{P}=0
\]
or equivalently 
\[
T_{SO_{3}(R)}=\{RQ,Q+Q^{\top}=0\}
\]
and the normal space is 
\[
N_{SO_{3}(R)}=\{RS,S-S^{\top}=0\}
\]
where $T_{SO_{3}(R)}$ and $N_{SO_{3}(R)}$ denote the tangent and
normal spaces at $R$ respectively.

ii. Furthermore, suppose at a point $R\in SO_{3}$, the gradient of
$f$ in $\mathbb{R}^{3\times3}$ is 
\[
\frac{\partial f}{\partial R}=M
\]
Then the projection of $M$ on the tangent space $T_{SO_{3}(R)}$
is given by 
\[
M_{T}=\frac{1}{2}M-\frac{1}{2}RM^{\top}R
\]
\end{lem} 
\begin{proof}
The first part of the proof regarding the tangent and normal spaces
of $SO_{3}$ (i.) is given in \cite{helmke1994optimization}. We now
prove that (ii.) the projection of $M$ on $T_{SO_{3}(R)}$ is given
by $M_{T}=\frac{1}{2}M-\frac{1}{2}RM^{\top}R$.

First, let $M_{N}=\frac{1}{2}M+\frac{1}{2}RM^{\top}R$. Observe that
\[
R^{\top}M_{N}=\frac{1}{2}R^{\top}M+\frac{1}{2}M^{\top}R
\]
is symmetric; therefore, $M_{N}\in N_{SO_{3}(R)}$ is normal to the
tangent space $T_{SO_{3}(R)}$ at $R$. Furthermore, 
\[
R^{\top}M_{T}=\frac{1}{2}R^{\top}M-\frac{1}{2}M^{\top}R
\]
is skew-symmetric; thus $M_{T}\in T_{SO_{3}(R)}$. Further because
$M_{T}+M_{N}=M$, $M_{T}$ is the projection of $M$ on $T_{SO_{3}(R)}$. 
\end{proof}
Now suppose $f=\sum_{i=1}^{N}(z[k]-\|Rp_{y}[k]+T-p_{x}[k]\|)^{2}$.
It is straightforward to obtain the gradient $M_{T}$ and $\frac{\partial f}{\partial T}$, and so
we have the gradient flow

\begin{equation}
\begin{split}\dot{R} & =-M_{T}\\
\dot{T} & =-\frac{\partial f}{\partial T}
\end{split}
\label{gradient flow}
\end{equation}
The SDP relaxation of Section \ref{sec:A tentative solution} is likely
to give a good initial condition, and then a gradient method using
a discretization of \eqref{gradient flow} is sufficient in solving
this optimization problem. The use of the Procrustes problem algorithm
can be applied in each step to tidy up departures from orthogonality
due to round-off and discretization error. 

\begin{rem}
In the case of multi-agent localization, the same algorithm can be adopted for those networks where each non-GPS equipped agent has 7 measurements to a GPS-equipped agent. Thus if there is only one GPS-equipped agent, the graph will be a star form; each of the GPS-denied agents (with INS) must have 7 or more distance measurements between it and the GPS-equipped agent. This allows a collection of two-agent problems to be solved (possibly in parallel). Note that this topology may well correspond to the physical situation of one agent remaining high above a building canyon, while other agents explore the building canyon. 
\end{rem}

\section{Simulations and Performance Evaluation}

\subsection{Simulation and Comparison between SDP and Gradient Descent}

\label{sec::simu} Simulations are given below. In the simulations,
white Gaussian noise is added to the measurements of $d$, and there are 7 such measurements. Furthermore, $\mbox{SNR}$
denotes the Signal-to-noise ratio in dB. Here SNR is defined as $\mathrm{SNR_{dB}}=20\log_{10}\left(\frac{A_{\mathrm{signal}}}{A_{\mathrm{noise}}}\right)$
where $A_{\mathrm{signal}}$ is the average distance between agents
and $A_{\mathrm{noise}}$ is the root mean square (RMS) amplitude of the noise added in to the measurements. 

\begin{itemize}
\item In each figure, the upper subplot shows the measurements $z$, the
second subplot shows the three entries of the translation $T$, and
the third subplot shows the nine entries of the rotation matrix $\bar{R}$.
The ordering of the entries is $r_{11},r_{12},r_{13},r_{21},\cdots$.
Note the correction as described in section \ref{sec::Procu Correction}
is applied. 
\item In the upper subplot, the blue line shows the true value of $z$ and
the red line shows the noisy distance observations $\tilde{z}$. 
\item In the middle and lower subplot, 'o' marks the true value, '+' marks
the value obtained by solving the SDP relaxation and '$\times$' marks
the value obtained by gradient descent refinement. 
\item Figure \ref{p1} depicts the noiseless case, i.e. $\mbox{SNR}\rightarrow\infty$,
in Figure \ref{p2} $\mbox{SNR}=30$ and in Figure \ref{p3} $\mbox{SNR}=10$.
It is well known that in practice the distance accuracy (1-sigma distance
error) is inversely proportional to signal bandwidth as well as the
square root of SNR. For passive RF detection problems the SNR in free
space is inversely proportional to distance squared. At 1km the SNR
is usually high say 20 or more dB. \footnote{The value is relevant for the DSTG
vehicles used to obtain the real data discussed below.}
\item Gradient descent refinement can provide an improved result
in comparison to simple SDP relaxation. This is partly because gradient
descent refinement uses a better estimator, and partly because of
the relaxation error in SDP. However, it is also notable that SDP
relaxation is important in providing the initial condition to avoid convergence to a 
local nonglobal minimum in the gradient descent search for the MLE optimum. 
\end{itemize}
\begin{figure}
\begin{centering}
\includegraphics[width=0.42\textwidth]{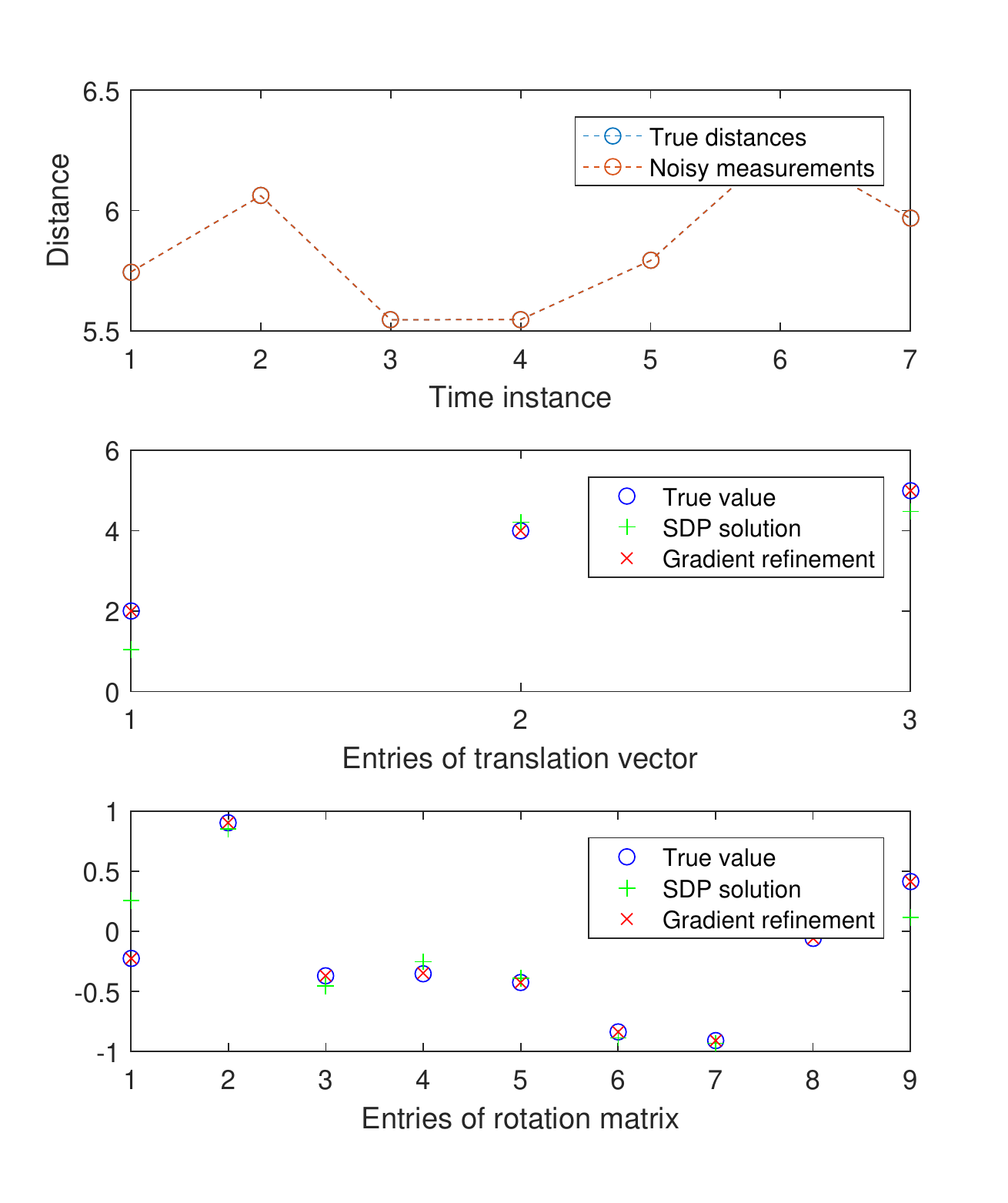} 
\par\end{centering}
\caption{Performance comparison between the SDP relaxation and gradient descent
refinement (noiseless)}
\label{p1} 
\end{figure}

\begin{figure}
\begin{centering}
\includegraphics[width=0.42\textwidth]{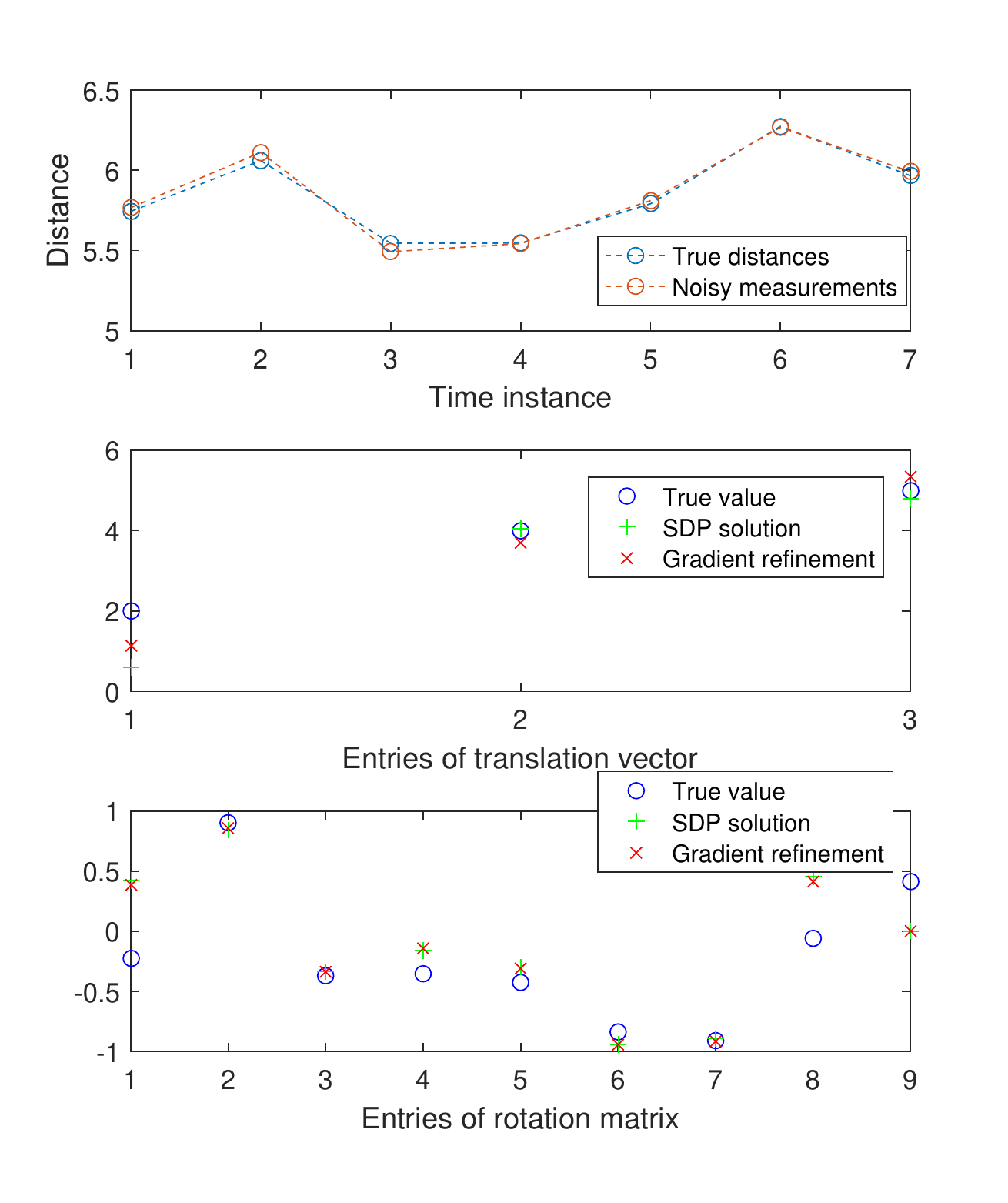} 
\par\end{centering}
\caption{Performance comparison between the SDP relaxation and gradient descent
refinement (SNR=30)}
\label{p2} 
\end{figure}

\begin{figure}
\begin{centering}
\includegraphics[width=0.42\textwidth]{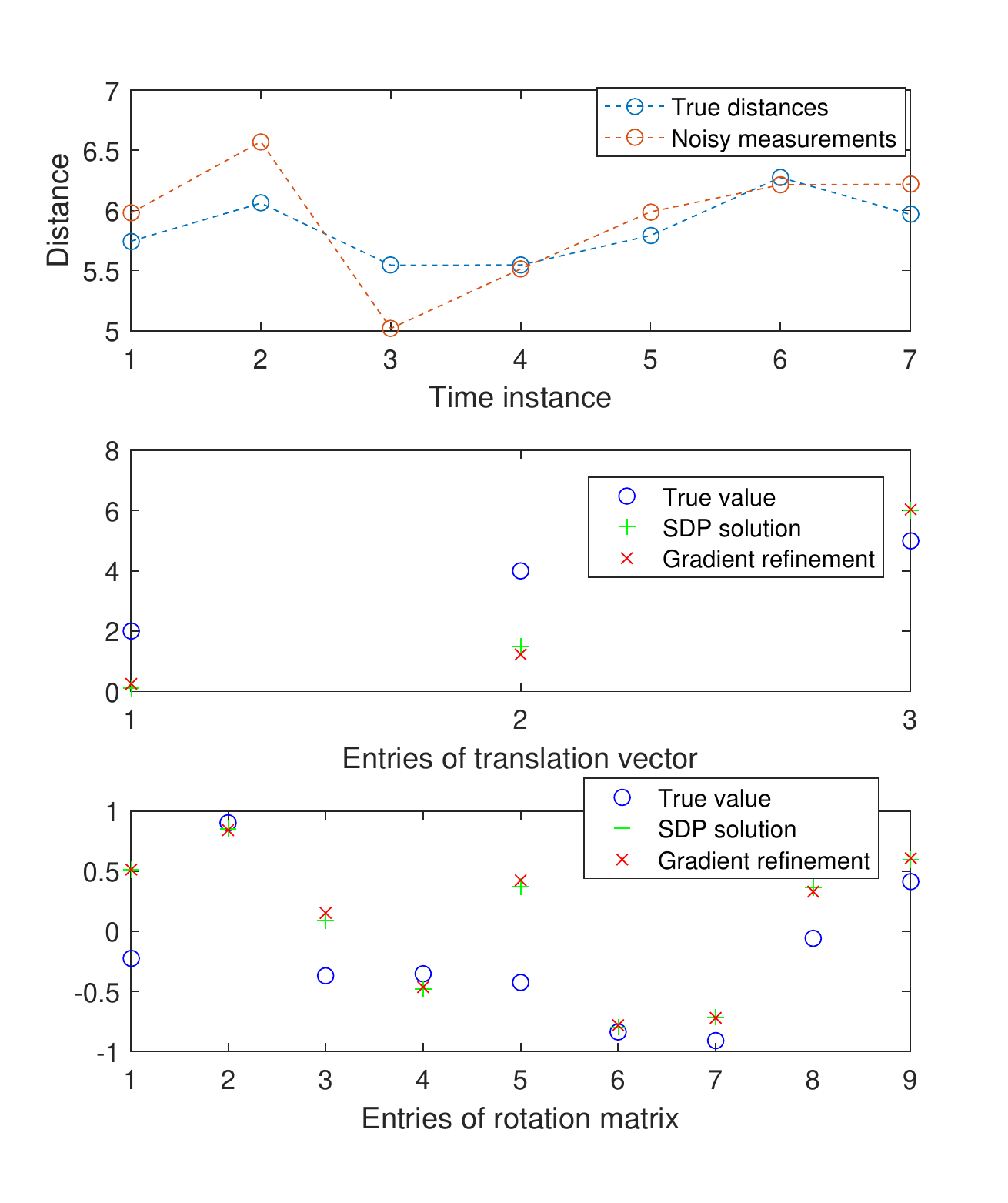} 
\par\end{centering}
\caption{Performance comparison between the SDP relaxation and gradient descent
refinement (SNR=10)}
\label{p3} 
\end{figure}

In the above simulations, which all have 7 measurements, a unique solution can be found for generic trajectories. A detailed argument based on algebraic geometry will then imply that for almost all instances of the problem,  the coefficient matrix $P$ in the objective function is of rank 7, carrying enough information to obtain a unique solution. Furthermore, as more and more measurements are obtained, the rank of $P$ will further increase, allowing more and more accurate estimation.   Nevertheless, there exist special cases where the rank of $P$ is always less than 7. 

Suppose for example the two agents pursue two parallel straight line paths. It is not hard to see there is insufficient data to obtain a unique rotation matrix. Also, with noise, nearly but not exactly parallel straight line paths may cause a problem.

With more measurements provided, the result can be improved
in both the SDP relaxation and gradient descent refinement for generic trajectories. 

\begin{rem} The objective function of the constrained linear least
squares problem and the objective function of the gradient descent
method are different. The first one is designed to exploit as far as possible the linear occurrence of unknowns in
 the $z^{2}$ expression while the second one is derived from a
maximum likelihood estimator. As noted, the constrained least squares
problem is used to find an appropriate starting point for the gradient
descent method, which then finds the accurate maximum likelihood estimate.
\end{rem}

\subsection{Number of measurements vs. estimation error}

\label{sec::number of measurements} From \cite{dietmaier1998stewart},
 one can conclude that with 6 distance measurements, the 3D geolocation
problem has 40 solutions, where some of the solutions may not be real and must be
discarded. The paper \cite{yu2006principles} further shows that an
additional measurement can almost always disambiguate the solutions
and with 7 distance-only measurements, a unique solution can be found. 

Although the use of SDP followed by the gradient-based MLE algorithm as proposed
in this paper does not require a minimum number of measurements, an
implicit assumption is that the number of measurements is greater
than or equal to 7. In the case of less than or equal to 5 measurements,
the gradient of the objective function of the MLE optimization is
always zero on a manifold, and the final result is randomly located
on the manifold. In addition, in the case of 6 measurements, there
is no way to deal with the issue of disambiguating local optima. With the number of measurements being greater than or equal to 7,
it is reasonable to expect that the performance of the approach will
be improved as the number of measurements increases. To study this, let us suppose a direction error measure is defined as 
\begin{equation}
E_{\mbox{direction}}=\mbox{arccos}\Big(\frac{\tilde{T}^{\top}T}{\|\tilde{T}\|\|T\|}\Big)\label{eq::relativ error angle}
\end{equation}

\begin{figure}
\begin{centering}
\includegraphics[width=0.42\textwidth]{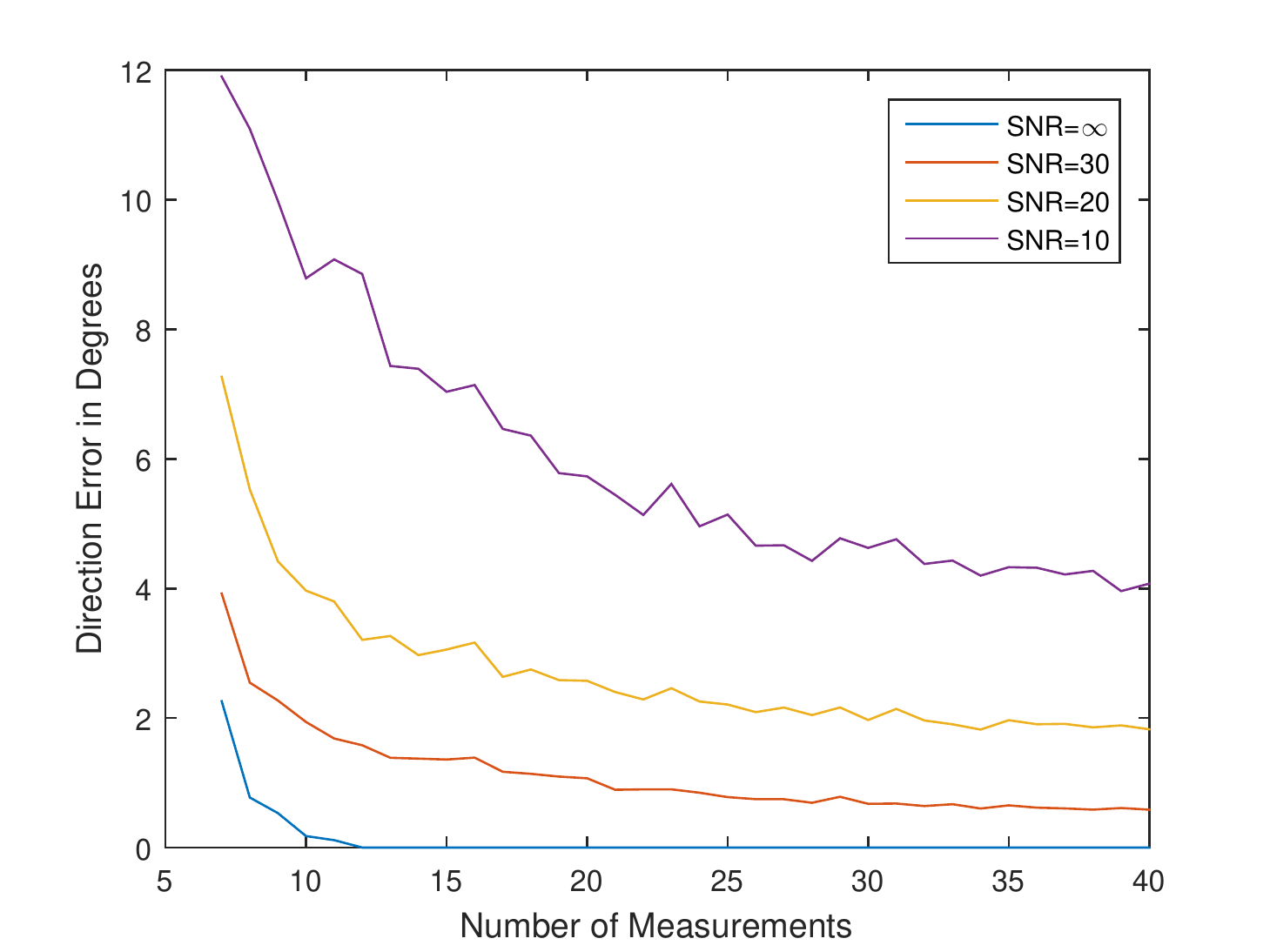} 
\par\end{centering}
\caption{Direction error in degrees vs. Number of measurements}
\label{fig::numMeasurePerformancedire} 
\end{figure}

Figure \ref{fig::numMeasurePerformancedire} shows the change of direction
error in degrees as the number of measurements increases. The simulations were run
with three different levels of SNR: 10dB, 20dB and 30 dB. In the above
figure, the blue curve shows the result with SNR$=$10, the red curve
shows the result with SNR$=$20 and the yellow curve shows the result
with SNR$=$30. With each number of measurements and each level of
SNR, the result shown in the above figures is the average of 200 simulations
with random vehicle trajectories and white noise.  

\subsection{Trial on real data}

Real flight data provided by Australian Defence Science and Technology Group is used in this section to test the performance
of the proposed method in practice. The data consists of: the true
positions of UAV1 in the global coordinates, the positions of UAV2
in its local INS coordinates and the distance measurements between the
pair of agents. The relevant numbers are recorded in Table \ref{my-label}.

\begin{table}[h]
\fontsize{6.5}{6.5}\selectfont \centering \caption{UAV localization trial}
\label{my-label}  %
\begin{tabular}{llllllll}
Time(s) & \multicolumn{3}{l}{UAV1 global Coordinates} & \multicolumn{3}{l}{UAV2 Local Coordinates} & Distances\tabularnewline
 & x(m)  & y(m)  & z(m)  & x(m)  & y(m)  & z(m)  & measured (m)\tabularnewline
3.2  & 349.1  & -924.1  & 374.4  & 1038.3  & 600.9  & 311.2  & 1541.0 \tabularnewline
12.9  & 576.2  & -945.0  & 371.9  & 1249.5  & 637.9  & 311.0  & 1516.7 \tabularnewline
22.7  & 781.0  & -870.3  & 372.5  & 1481.7  & 679.0  & 308.5  & 1381.6 \tabularnewline
32.5  & 936.8  & -712.4  & 373.7  & 1708.6  & 717.2  & 309.4  & 1149.9 \tabularnewline
42.3  & 1007.0  & -522.7  & 373.3  & 1939.5  & 755.1  & 309.3  & 907.7 \tabularnewline
52.1  & 992.0  & -299.4  & 373.5  & 2084.3  & 867.7  & 308.5  & 800.7 \tabularnewline
61.8  & 869.8  & -91.3  & 373.2  & 2040.9  & 1088.0  & 309.4  & 899.1 \tabularnewline
71.6  & 660.1  & 38.6  & 372.9  & 2004.3  & 1305.7  & 310.1  & 1120.9 \tabularnewline
80.5  & 431.4  & 56.6  & 373.1  & 1976.8  & 1507.0  & 309.5  & 1435.9 \tabularnewline
91.1  & 189.7  & -49.1  & 373.3  & 1933.3  & 1723.2  & 310.7  & 1867.1 \tabularnewline
100.9  & 33.9  & -262.2  & 373.6  & 1724.3  & 1755.4  & 310.5  & 2084.6 \tabularnewline
\end{tabular}
\end{table}

Now using the above constrained least-squares method followed by gradient
optimization, one can compute UAV 2's positions in the global coordinate
system as shown in Figure \ref{fig::real}. In Figure \ref{fig::real},
circles are the positions of UAV 1 in global coordinates; triangles
are the recovered position of UAV 2 in the global coordinates, and
the solid line is the true trajectory of UAV2. Comparing the triangles
and the solid line, we observe that the localization algorithm achieves
better accuracy in the north and east directions but poorer accuracy
in height. This can be explained by the following observation: we
find that the trajectories of the UAVs are very close to coplanar,
which may be detrimental to the reliability of the result. Indeed,
if we change the height recorded by a UAV's inertial sensor by a few
metres, a 'mirror solution' will be obtained. This statement can be
drawn from an analysis of the (reduced) rigidity matrix, see \cite{anderson2010formal}.

\begin{figure}
\begin{centering}
\includegraphics[width=0.42\textwidth]{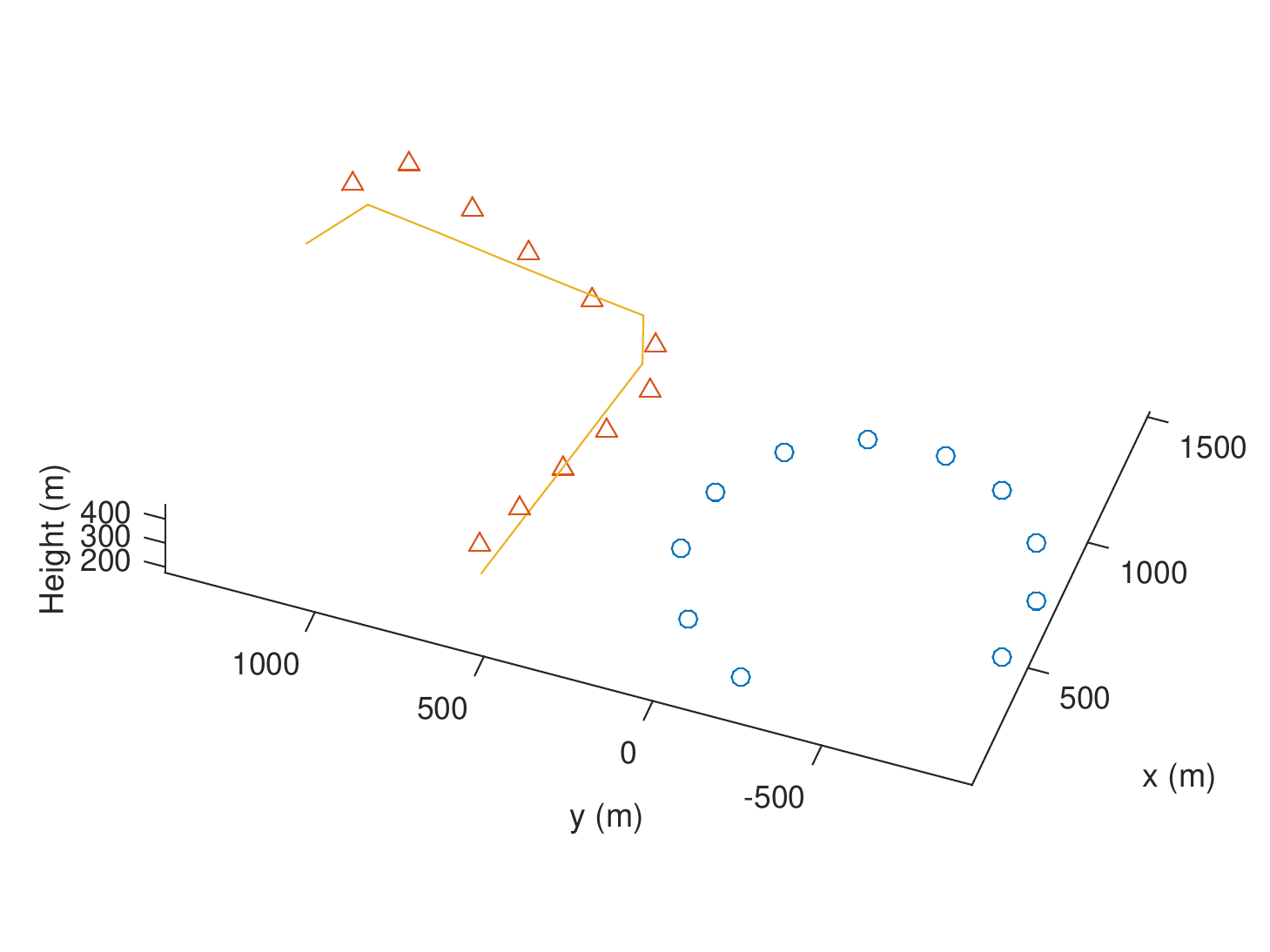} 
\par\end{centering}
\caption{3D plot of the UAV positions in global coordinate systems}
\label{fig::real} 
\end{figure}

\section{Conclusions}

In this paper, we first proposed a novel semidefinite optimization
approach for solving the problem of 3D mobile localization of a GPS-denied agent using distance-only
measurements. After that, a maximum likelihood estimator is used in
a further approach to enhance the accuracy of localization, with simulations
using real field test data. 

Future work includes introducing systematic treatment for the multi-agent case using similar procedure with this paper.
We are also involved in a separate study of localization of GPS-denied
agents using bearing-only (azimuth and elevation) measurements. 

\section*{Appendix}
The 10 constraints are listed as below. Let $C_{i}(\Theta)$ be the $i$th constraint and
we have 
\begin{equation}
\begin{split}C_{1} & =\theta_{1}^{2}+\theta_{2}^{2}+\theta_{3}^{2}-1=0\\
C_{2} & =\theta_{4}^{2}+\theta_{5}^{2}+\theta_{6}^{2}-1=0\\
C_{3} & =\theta_{7}^{2}+\theta_{8}^{2}+\theta_{9}^{2}-1=0\\
C_{4} & =\theta_{1}^{2}+\theta_{4}^{2}+\theta_{7}^{2}-1=0\\
C_{5} & =\theta_{2}^{2}+\theta_{5}^{2}+\theta_{8}^{2}-1=0\\
C_{6} & =\theta_{1}\theta_{2}+\theta_{4}\theta_{5}+\theta_{7}\theta_{8}=0\\
C_{7} & =\theta_{1}\theta_{10}+\theta_{4}\theta_{11}+\theta_{7}\theta_{12}-\theta_{13}=0\\
C_{8} & =\theta_{2}\theta_{10}+\theta_{5}\theta_{11}+\theta_{8}\theta_{12}-\theta_{14}=0\\
C_{9} & =\theta_{10}^{2}+\theta_{11}^{2}+\theta_{12}^{2}-\theta_{16}=0\\
C_{10} & =\theta_{13}^{2}+\theta_{14}^{2}+\theta_{15}^{2}-\theta_{16}=0
\end{split}
\label{eq::constraints}
\end{equation}

\textcolor{black}{Note there are 10 independent equality constraints
and 16 independent variables, so the problem has 6 degrees of freedom.
That is consistent with the fact that each of the rotation matrix
$R$ and the translation matrix $T$ has three degrees of freedom.
One should also note that the equation set used to express those constraints
is not unique. In fact, there are 4 other constraints being dropped
here, which can be derived from $C_{i}=0,i=1,\cdots,10$. The additional
constraints are}

\textcolor{black}{
\[
\begin{split}C_{11} & =\theta_{3}^{2}+\theta_{6}^{2}+\theta_{9}^{2}-1=0\\
C_{12} & =\theta_{1}\theta_{3}+\theta_{4}\theta_{6}+\theta_{7}\theta_{9}=0\\
C_{13} & =\theta_{2}\theta_{3}+\theta_{5}\theta_{6}+\theta_{8}\theta_{9}=0\\
C_{14} & =\theta_{3}\theta_{10}+\theta_{6}\theta_{11}+\theta_{9}\theta_{12}-\theta_{15}=0
\end{split}
\]
}

\bibliographystyle{plain} 
\bibliography{reference}

\begin{thebibliography}{10}

\bibitem{anderson2010formal}
Brian~DO Anderson, Iman Shames, Guoqiang Mao, and Baris Fidan.
\newblock Formal theory of noisy sensor network localization.
\newblock {\em Journal on Discrete Mathematics}, 24(2):684--698, 2010.

\bibitem{cao2008reaching}
Ming Cao, A~Stephen Morse, and Brian~DO Anderson.
\newblock Reaching a consensus in a dynamically changing environment: a
  graphical approach.
\newblock {\em SIAM Journal on Control and Optimization}, 47(2):575--600, 2008.

\bibitem{cao2011formation}
Ming Cao, Changbin Yu, and Brian~DO Anderson.
\newblock Formation control using range-only measurements.
\newblock {\em Automatica}, 47(4):776--781, 2011.

\bibitem{chepuri2014rigid}
Sundeep~Prabhakar Chepuri, Geert Leus, and Alle-Jan van~der Veen.
\newblock Rigid body localization using sensor networks.
\newblock {\em IEEE Transactions on Signal Processing}, 62(18):4911--4924,
  2014.

\bibitem{Chung2005coupler}
Wen-Yeuan Chung.
\newblock The characteristics of a coupler curve.
\newblock {\em Mechanism and Machine Theory}, 40(10):1099 -- 1106, 2005.

\bibitem{connelly2013generic}
Robert Connelly, Tibor Jord{\'a}n, and Walter Whiteley.
\newblock Generic global rigidity of body--bar frameworks.
\newblock {\em Journal of Combinatorial Theory, Series B}, 103(6):689--705,
  2013.

\bibitem{dietmaier1998stewart}
Peter Dietmaier.
\newblock The stewart-gough platform of general geometry can have 40 real
  postures.
\newblock In {\em Advances in Robot Kinematics: Analysis and Control}, pages
  7--16. Springer, 1998.

\bibitem{gower2004procrustes}
John~C Gower and Garmt~B Dijksterhuis.
\newblock {\em Procrustes problems}, volume~3.
\newblock Oxford University Press Oxford, 2004.

\bibitem{helmke1994optimization}
Uwe Helmke, John~B Moore, and Wrzburg Germany.
\newblock {\em Optimization and dynamical systems}.
\newblock Springer, 1994.

\bibitem{hopcroft2012computer}
John Hopcroft and Ravi Kannan.
\newblock Computer science theory for the information age.
\newblock 2012.

\bibitem{jiang2015simultaneous}
B.~Jiang, M.~Deghat, and B.~Anderson.
\newblock Simultaneous velocity and position estimation via distance-only
  measurements with application to multi-agent system control.
\newblock {\em IEEE Transactions on Automatic Control}, PP(99):1--1, 2016.

\bibitem{jiang2013translational}
Bomin Jiang, M.~Deghat, and B.D.O. Anderson.
\newblock Translational velocity consensus using distance-only measurements.
\newblock In {\em IEEE Conference on Decision and Control}, pages 2746--2751,
  Dec 2013.

\bibitem{chapter9}
Usman Mansoor and Habib~M. Ammari.
\newblock Localization in three-dimensional wireless sensor networks.
\newblock In Habib~M. Ammari, editor, {\em The Art of Wireless Sensor
  Networks}, chapter~9, pages 325--363. Springer, 2014.

\bibitem{mao2007wireless}
Guoqiang Mao, Bar{\i}{\c{s}} Fidan, and Brian~DO Anderson.
\newblock Wireless sensor network localization techniques.
\newblock {\em Computer networks}, 51(10):2529--2553, 2007.

\bibitem{motevallian2013new}
S.A. Motevallian, Lu~Xia, and B.D.O. Anderson.
\newblock A new splitting-merging paradigm for distributed localization in
  wireless sensor networks.
\newblock In {\em IEEE International Conference on Communications}, pages
  1454--1458, June 2013.

\bibitem{nagata2002multi}
Takeshi Nagata and Hiroshi Sasaki.
\newblock A multi-agent approach to power system restoration.
\newblock {\em IEEE Transactions on Power Systems}, 17(2):457--462, 2002.

\bibitem{you2017underwater}
Keyou~You Qizhu~Chen and Shiji Song.
\newblock Cooperative localization for autonomous underwater vehicles using
  parallel projection.
\newblock In {\em 13th IEEE International Conference on Control and Automation,
  to appear}, 2017.

\bibitem{shamsi2010sensor}
Davood Shamsi, Nicole Taheri, Zhisu Zhu, and Yinyu Ye.
\newblock On sensor network localization using sdp relaxation.
\newblock {\em arXiv preprint arXiv:1010.2262}, 2010.

\bibitem{stewart1966plantform}
D.~Stewart.
\newblock A platform with six degrees of freedom.
\newblock {\em Aircraft Engineering and Aerospace Technology}, 38(4):30--35,
  1966.

\bibitem{tay1984recent}
Tiong-Seng Tay and Walter Whiteley.
\newblock Recent advances in the generic ridigity of structures.
\newblock {\em Structural Topology, 1984, n{\'u}m. 9}, 1984.

\bibitem{trawny20073d}
Nikolas Trawny, Xun~S Zhou, Ke~X Zhou, and Stergios~I Roumeliotis.
\newblock 3d relative pose estimation from distance-only measurements.
\newblock In {\em IEEE/RSJ International Conference on Intelligent Robots and
  Systems}, pages 1071--1078. IEEE, 2007.

\bibitem{yu2006principles}
Changbin Yu, Bar{\i}s Fidan, and Brian~DO Anderson.
\newblock Principles to control autonomous formation merging.
\newblock In {\em Proceedings of the American Control Conference}, pages 7--pp.
  IEEE, 2006.

\end{thebibliography}

\end{document}